\documentclass[aps,amsmath,amssymb,twocolumn,pra,nofootinbib,superscriptaddress]{revtex4-2}

\usepackage{graphicx}
\usepackage{dcolumn}
\usepackage{bm}
\usepackage{epsfig,hyperref,amsthm}
\usepackage{mathtools}
\usepackage{color}
\usepackage{colordvi}
\usepackage[dvipsnames]{xcolor}
\usepackage{relsize}
\usepackage{accents}
\usepackage{wasysym}  
\usepackage{xypic}
\usepackage{times}
\usepackage{newtxmath}
\usepackage{amsfonts}

\hypersetup{
	colorlinks=true,
	linkcolor=CitingColor,
	citecolor=CitingColor,
	urlcolor=CitingColor
}

\DeclareMathAlphabet\mathbfcal{OMS}{cmsy}{b}{n}
\newcommand{\ket}[1]{\ensuremath{|#1\rangle}}
\newcommand{\bra}[1]{\ensuremath{\langle #1|}}

\newcommand{\proj}[1]{\ket{#1}\bra{#1}}
\newcommand{\be}{\begin{equation}}
\newcommand{\ee}{\end{equation}}
\newcommand{\ba}{\begin{eqnarray}}
\newcommand{\ea}{\end{eqnarray}}

\newcommand{\norm}[1]{\left\|#1\right\|}

\newcommand{\id}{\mathbb{I}}

\newcommand{\wt}[1]{\widetilde{#1}}

\newcommand{\mE}{\mathcal{E}}

\newcommand{\Sin}{S}
\newcommand{\Sout}{S'}
\newcommand{\Wdiff}{\Delta}
\newcommand{\norsigma}{{\color{black}\eta}}
\newcommand{\assemsigma}{{\color{black}\mathcal{A}}}
\newcommand{\assemtau}{{\color{black}\mathcal{B}}}

\newcommand{\assem}{
{\color{black}{\mathbfcal{A}}}
}

\newcommand{\assemLHS}{
{\color{black}{\mathbfcal{B}}}
}

\newtheorem{result}{Theorem}
\newtheorem{result-coro}[result]{Corollary}

\newtheorem{question}{Question}

\definecolor{nred}{rgb}{0.9,0.1,0.1}
\definecolor{nblack}{rgb}{0,0,0}
\definecolor{nblue}{rgb}{0.2,0.2,0.8}
\definecolor{ngreen}{rgb}{0.2,0.6,0.2}
\definecolor{ublue}{rgb}{0,0,0.5}

\definecolor{pur}{rgb}{0.75,0,0.75}
\definecolor{nngrn}{rgb}{0,0.5,0.5}
\definecolor{CitingColor}{rgb}{0,0.3,1}

\newcommand{\blu}{\color{nblue}}

\newcommand{\CY}[1]{{#1}}
\newcommand{\CYtwo}[1]{{#1}}
\newcommand{\CYthree}[1]{{#1}}

\begin{document}
\title{
General quantum resources providing advantages in work-extraction tasks
}

\author{Chung-Yun Hsieh}
\email{chung-yun.hsieh@bristol.ac.uk}
\affiliation{H. H. Wills Physics Laboratory, University of Bristol, Tyndall Avenue, Bristol, BS8 1TL, United Kingdom}

\author{Manuel Gessner}
\email{manuel.gessner@uv.es}
\affiliation{Instituto de Física Corpuscular (IFIC), CSIC‐Universitat de València and Departament de Física Teòrica, UV, Avinguda Vicent Andrés Estellés 19, E-46100 Burjassot (Valencia), Spain} 

\date{\today}

\begin{abstract}
In developing quantum science and technologies, it is essential to demonstrate so-called {\em quantum advantages}, which are performances that can be achieved only with the assistance of quantum resources.
Most of the time, different quantum features lead to different advantages. 
Interestingly, there are certain classes of tasks where quantum advantages are achievable by {\em general} quantum resources.
This work reports such a class of tasks in thermodynamics---we provide a work-extraction task that certifies general quantum resources of both states and channels, suggesting {\em general} quantum effects can provide non-classical advantages in work extraction.
We also show that such work-extraction tasks can be applied to certify quantum entanglement in a {\em one-sided device-independent} way.
As an application, we report a novel type of anomalous energy flow---a type of locally extractable energy that is attributed to the globally distributed entanglement. 
Finally, we show that the existence of this novel anomalous energy flow is equivalent to measurement incompatibility.
\end{abstract}

\maketitle

\section{Introduction}
\CY{{\em Quantum advantages} underpin quantum science and technologies.
They are tasks and benchmarks unachievable by any classical means}~\cite{Kuroiwa2024PRL,Kuroiwa2024PRA,Meier2025PRXEnergy}.
For instance, quantum teleportation~\cite{Bennett93}, super-dense coding~\cite{Bennett92}, and sub-shot-noise interferometric precision with qubit probes~\cite{PhysRevLett.102.100401,RevModPhys.90.035005,PhysRevLett.126.080502} are possible only when we consume quantum entanglement as a resource~\cite{HorodeckiRMP}. 
On the other hand, in device-independent quantum information processing, quantum nonlocality~\cite{Brunner2014RMP,Acin2007PRL}, quantum steering~\cite{UolaRMP2020,Cavalcanti2016,Branciard2012PRA,Ku2022NC,Hsieh2023-3}, quantum incompatibility~\cite{Otfried2021Rev,Ku2023}, and quantum complementarity~\cite{Hsieh2023} can enhance the performance in cryptography and communication.
Furthermore, dynamical quantum effects can also be resources for, e.g., quantum memories~\cite{Rosset2018PRX,Yuan2021npjQI,Ku2022PRXQ,Vieira2024,Abiuso2024,NarasimhacharPRL2019,Hsieh2025PRA-3}, quantum communication~\cite{Takagi2020PRL,Hsieh2021PRXQ,Hsieh2025PRL,Hsieh2025PRA}, preserving or generating quantumness~\cite{Hsieh2020,Hsieh2021PRXQ,Liu2019DRT,Liu2020PRR,Hsieh2020PRR,Streltsov2015PRL,Hsieh2025PRA-3}, and thermodynamics~\cite{Hsieh2021PRXQ,Hsieh2025PRL,Hsieh2025PRA,Stratton2023}.

Crucially, most of the quantum advantages mentioned above are only possible with {\em specific} resources.
It is thus of great interest, both theoretically and practically, to know whether there is any class of tasks where {\em general} quantum resources can provide an advantage.
\CY{This amounts to asking:
\begin{center}
{\em Can any class of tasks universally certify quantum resources?}
\end{center}}
We call such a class a {\em universal resource certification class} (URCC).
Whenever a URCC is found, \CYtwo{being able} to perform it experimentally is then equivalent to the ability to experimentally certify general quantum resources.
Consequently, discovering a URCC in a novel setting is not only of theoretical interest, but may also provide a novel method to practically certify a broad range of quantum phenomena.

Using a theoretical approach termed {\em quantum resource theories}~\cite{ChitambarRMP2019}, or simply {\em resource theories}, it has been proved that general quantum resources\footnote{With general quantum resources, \CY{we refer to} all resources with reasonable and minimal physical assumptions, including convexity and compactness of the free set; \CY{see also Ref.~\cite{Companion-letter}.}} can provide advantages in the classes of operational tasks such as discrimination~\cite{Piani2009PRL,Takagi2019PRL,Skrzypczyk2019PRL,Bae2019PRL,Takagi2019,Skrzypczyk2019,Hsieh2023-2,Hsieh2022PRR}, exclusion~\cite{Ducuara2020PRL,Hsieh2023,Uola2020PRL}, parameter-estimation in metrology~\cite{Tan2021PRL}, \CY{and input-output games~\cite{Uola2020PRA} (see also Ref.~\cite{Rosset2020PRL} for spacelike separated resources). 
These existing URCCs thus tell us not just how to universally certify quantum resources---they also tell us that general quantum resources can provide {\em generic} advantages in these tasks.}

Surprisingly, despite several URCCs having been reported, there is no known URCC in the context of thermodynamics.
So far, quantum signatures in thermodynamics have been identified in, e.g., quantum heat engines~\cite{KosloffPRE2002,FeldmannPRE2006,JiPRL2022,BeyerPRL2019,ChanPRA2022,Biswas2025PRL,Hsieh2024} and conservation laws~\cite{LostaglioNJP2017,Majidy2023,YungerHalpern2016NC,Guryanova2016NC} (see also Refs.~\cite{Lostaglio2020PRL,Levy2020PRXQ,Puliyil2022PRL,Upadhyaya2023,Centrone2024}). 
Also, quantum advantages in thermodynamics of specific resources such as entanglement~\cite{Lipkabartosik2023,Perarnau-LlobetPRX2015,Jennings2010PRE,Rio2011,Skrzypczyk2014NC}, coherence~\cite{Shiraishi2023,Korzekwa2016NJP}, steering~\cite{JiPRL2022,BeyerPRL2019,ChanPRA2022,Biswas2025PRL}, incompatible instruments~\cite{Hsieh2024}, and quantum communication channels~\cite{Hsieh2020PRR,Hsieh2021PRXQ,Hsieh2020,Hsieh2025PRL,Hsieh2025PRA} have been studied. 
\CY{Furthermore, very recently, {\em heat} was found to serve as a witness of certain quantum properties~\cite{deOliveiraJunior2025PRL}.}
Still, a task-based thermodynamic URCC has not been found to date, and it is still unknown whether general quantum resources can provide quantum advantages in a single class of thermodynamic tasks.

In this work, we provide the first thermodynamic URCC based on work-extraction tasks, \CYtwo{where general quantum resources can provide advantages.}
\CY{Moreover, in the context of quantum steering~\cite{EPR,Schrodinger1935,Schrodinger1936,Wiseman2007PRL,Jones2007PRA,UolaRMP2020,Cavalcanti2016}, our results not only show that entanglement can be certified by work extraction in a one-sided device-independent way, but also lead to a novel type of anomalous energy flow~\cite{Lipkabartosik2023,Jennings2010PRE}, which is locally extracted and can only be generated by the combination of entanglement and incompatible measurements~\cite{Otfried2021Rev}.
 
This work is the companion paper of Ref.~\cite{Companion-letter}, which \CYtwo{provides a necessary and sufficient characterisation of state conversions in general resource theories via thermodynamic work extraction.}
Here, we provide additional details on the work-extraction task, discuss the implications for specific resources, and extend our results to enable the certification of resources beyond states, i.e., of channels and \CYthree{channel tuples (which refers to finite collections of channels)}.}

\CY{We structure this work as follows.}
In Sec.~\ref{Sec:preliminary}, we detail the notion of \CYthree{channel tuples} and their resource theories.
In Sec.~\ref{Sec:Discrimination}, we present the first set of main results, which shows that the so-called {\em ensemble state discrimination tasks} introduced in Ref.~\cite{Hsieh2022PRR} form a URCC for \CYthree{channel tuple} resource theories (Theorems~\ref{Thm:UolaTheorem} and~\ref{Result:ImplicationDynamicalResourceTheory}).
In Sec.~\ref{Sec:Work extraction}, we present the second set of main results, which show that the work-extraction tasks introduced in this paper form a URCC for \CYthree{channel tuple} resource theories (Theorem~\ref{AppResult:ThermoOperationalInterpretation}).
\CY{In Sec.~\ref{Sec:1SDI}, by focusing on the quantum steering scenarios~\cite{EPR,Schrodinger1935,Schrodinger1936,Wiseman2007PRL,Jones2007PRA,UolaRMP2020,Cavalcanti2016}, we develop a one-sided device-independent work-extraction task that identifies a local anomalous energy flow (Theorem~\ref{Result:Steering}). Observation of this phenomenon can be used to certify entanglement and measurement incompatibility~\cite{Otfried2021Rev,UolaRMP2020}.
Finally, we conclude the paper in Sec.~\ref{Sec:Conclusion}.
}

\section{Preliminary Notions}\label{Sec:preliminary}

\subsection{Resource \CYthree{Theories} of \CYthree{Channel Tuples}}
We focus on finite-dimensional quantum systems where a {\em  quantum state}, or simply {\em state}, $\rho$, is a positive semi-definite operator $\rho\ge0$ with unit trace ${\rm tr}(\rho)=1$. 
{\em Quantum measurements} are modelled by {\em positive operator-valued measures} (POVMs)~\cite{QIC-book}, that are sets of positive semi-definite operators $\{E_i\}_i$ satisfying $E_i\ge0$ $\forall\,i$ and $\sum_iE_i=\id$, where $\id$ is the identity operator acting on the given system.
For an initial state $\rho$, a POVM $\{E_i\}_i$ describes the measurement process that outputs the measurement outcome indexed by $i$ with the probability ${\rm tr}(E_i\rho)$.
Finally, the general evolution of quantum systems is described by a {\em quantum channel}, or simply {\em channel}, that is a {\em completely-positive trace-preserving} linear map~\cite{QIC-book}.
Physically, a channel $\mathcal{E}$ mapping from an input system $\Sin$ to an output system $\Sout$ describes the process that transforms any given initial state $\rho$ in $\Sin$ to the final state $\mathcal{E}(\rho)$ in $\Sout$.

\CY{We can now define a key notion for this work.
Consider a given input system $\Sin$ and output system $\Sout$.
A \CYthree{\em channel tuple} is defined as a finite collection \CYthree{of channels, denoted by}
\begin{align}
\mathbfcal{E}\coloneqq\{\mathcal{E}_x:{\Sin_x\to\Sout_x}\}_{x=1}^L,
\end{align}
where $\mathcal{E}_x$ is a channel from $\Sin_x$ to $\Sout_x$, and, for each $x$, $\Sin_x$ ($\Sout_x$) is a subsystem of $\Sin$ ($\Sout$).
Throughout this work, we will call ``$\Sin_x\to\Sout_x$'' an {\em input-output route}.}
\CYthree{For clarity, a {\em fixed} set of input-output routes will be denoted as}
\begin{align}
{\bm\Lambda}\coloneqq\{\Sin_x\to\Sout_x\}_{x=1}^L, 
\end{align}
and we aim to identify \CYthree{channel tuples (with the given ${\bm\Lambda}$)} that possess a resource of interest (denoted by ``$R$'').
To this end, following the standard resource-theoretic approach, we first identify the set of all \CYthree{channel tuples} {\em without} this resource---this is called the {\em free set} of the resource $R$ and is written as $\mathfrak{F}_R$.
\CYthree{Channel tuples} in $\mathfrak{F}_R$ are called {\em free} or {\em $R$-resourceless}, and a \CYthree{channel tuple} $\mathbfcal{E}$ is called {\em $R$-resourceful} if $\mathbfcal{E}\notin\mathfrak{F}_R$.

\subsection{Examples}
The notion of \CYthree{channel tuples} encompasses a wide range of resource theories of channels and states. For instance, when we set ${\bm\Lambda} = \{\Sin\to\Sin'\}$, which contains only one input-output route, we recover the so-called {\em channel resource theories}~\cite{Liu2019DRT,Liu2020PRR}, including
\begin{itemize}
\item {\em Quantum memory.} 
\CYthree{This resource is described by} setting $\mathfrak{F}_R$ as the set of all entanglement-breaking channels~\cite{Rosset2018PRX,Yuan2021npjQI,Tabia2024,Ku2022PRXQ,Horodecki2003}, \CYthree{which are the measure-and-prepare channels~\cite{Horodecki2003} of the form 
\begin{align}\label{Eq:ent breaking channel}
(\cdot)\mapsto\sum_i\rho_i{\rm tr}[E_i(\cdot)]
\end{align} 
for some states $\rho_i$ and POVM $\{E_i\}_i$}. 
\item {\em Classical communication.} This can be obtained by setting $\mathfrak{F}_R$ as the set of all state-preparation channels~\cite{Takagi2020PRL} \CYthree{of the form 
\begin{align}\label{Eq:state-preparation channel}
(\cdot)\mapsto\rho{\rm tr}(\cdot)
\end{align} 
for a fixed state $\rho$.
Note that this is a special type of measure-and-prepare channel given in Eq.~\eqref{Eq:ent breaking channel}.}
\item {\em State resource theories.} By setting $\mathfrak{F}_R$ as the set of all state-preparation channels \CYthree{[Eq.~\eqref{Eq:state-preparation channel}]} outputting some free state of a given state resource, \CYthree{which} is denoted by ``$R_{\rm state}$,'' one can recover the state resource theory of $R_{\rm state}$ by restricting to state-preparation channels.
\item \CYthree{{\em Resource preservability.} One can set $\mathfrak{F}_R$ as the appropriate resource-annihilating channels~\cite{Hsieh2020,Hsieh2021PRXQ,Stratton2023,Hsieh2025PRA-3}.
More precisely, for a given state resource $R_{\rm state}$ with free set $\mathfrak{F}_{R_{\rm state}}$ (viewed as a set of states now), a resource-annihilating channel for $R_{\rm state}$ can be defined as a channel $\mathcal{N}$ satisfying $\mathcal{N}(\rho)\in\mathfrak{F}_{R_{\rm state}}$ $\forall\,\rho$~\cite{Hsieh2020}.}
\end{itemize}

To further illustrate the flexibility of \CYthree{channel tuples}, consider the case of $\Sin=A,\Sout=BC$ and ${\bm\Lambda}=\{A\to B,A\to C\}$. 
In this case, we can obtain the following resource \CYthree{theory}:
\begin{itemize}
\item {\em Broadcasting incompatible channels.} We set $\mathfrak{F}_R$ as the set of all broadcast compatible channels~\cite{Haapasalo2014,Girard2021NC,Viola2015, Haapasalo2021Quantum,Haapasalo2015,Carmeli2019PRL,Carmeli2019JMP}, \CYthree{which are channel tuples of two channels} \mbox{$\mathbfcal{E} = \{\mathcal{E}_1:{A\to B},\mathcal{E}_2:{A\to C}\}$} \CYthree{satisfying \begin{align}\label{Eq:broadcast compatible}
\mathcal{E}_1 = {\rm tr}_C\circ\mathcal{N}\quad\&\quad\mathcal{E}_2 = {\rm tr}_B\circ\mathcal{N}
\end{align} 
for at least one channel $\mathcal{N}:{A\to BC}$.}
\end{itemize}
\CYthree{If now we consider $\Sout=ABC$ and \mbox{${\bm\Lambda}=\{\Sin\to AB,\Sin\to BC\}$}, we obtain transitivity and other quantum marginal problems:
\begin{itemize}
\item {\em Quantum marginal problems and transitivity problems.} By restricting to only state-preparation channels [Eq.~\eqref{Eq:state-preparation channel}] in this setting, we can encompass various quantum marginal problems as well as transitivity problems~\cite{Hsieh2023-2,Tabia2022npjQI,Chen2025,Liu2025,Tabia2024}.
For instance, for two bipartite states $\rho_{AB}$ and $\rho_{BC}$ (in $AB$ and $BC$, respectively), a quantum marginal problem asks whether there is at least one tripartite state $\eta_{ABC}$ achieving 
\begin{align}\label{Eq:compatible states}
{\rm tr}_A(\eta_{ABC})=\rho_{BC}\quad\&\quad{\rm tr}_{C}(\eta_{ABC})=\rho_{AB}.
\end{align}
Then, for $\rho_{AB}$ and $\rho_{BC}$ with Eq.~\eqref{Eq:compatible states} for at least one $\eta_{ABC}$, the entanglement (meta-)transitivity problem~\cite{Tabia2022npjQI} asks whether {\em every} $\eta_{ABC}$ satisfying Eq.~\eqref{Eq:compatible states} must have an entangled marginal ${\rm tr}_{B}(\eta_{ABC})$ in $AC$.    
Our framework can include these problems by restricting to only state-preparation channels [Eq.~\eqref{Eq:state-preparation channel}].
\end{itemize}
}

Finally, when we set \CYthree{$\Sin=\Sout=ABC$ with $\Sin_1=\Sout_1=AB$ and $\Sin_2=\Sout_2=BC$, we obtain ${\bm\Lambda} = \{AB\to AB,BC\to BC\}$.}
In this setting, we can recover \CYthree{a special type of} the resource theory of 
\begin{itemize}
\item {\em Quantum channel marginal problems.} One can achieve this by setting $\mathfrak{F}_R$ as the set of \CYthree{compatible channels.
More precisely, two bipartite channels \mbox{$\mathcal{E}_1:AB\to AB$} and \mbox{$\mathcal{E}_2:BC\to BC$} are said to be compatible~\cite{Hsieh2022PRR} if there exists at least one tripartite channel \mbox{$\mathcal{N}:ABC\to ABC$} such that
\begin{align}
{\rm tr}_C\circ\mathcal{N}=\mathcal{E}_1\circ{\rm tr}_C\quad\&\quad{\rm tr}_A\circ\mathcal{N}=\mathcal{E}_2\circ{\rm tr}_A.
\end{align}
Note that this is similar to, but not the same as, the broadcast compatible condition defined in Eq.~\eqref{Eq:broadcast compatible}.
See Definitions 2 and 3 in Ref.~\cite{Hsieh2022PRR} for details.}
\end{itemize}

Consequently, by studying the resource theory of \CYthree{channel tuples}, we are able to obtain results that can naturally applied to a broad range of different quantum resources.

\section{Certifying quantum resources by discrimination tasks}\label{Sec:Discrimination}

\subsection{Certifying General Quantum Resources}
\CY{We start with a necessary and sufficient certification of} quantum resources possessed by \CYthree{channel tuples}, which further generalises Ref.~\cite{Uola2020PRA}'s main result \CY{(since now, with a given ${\bm\Lambda}$ and a given resource $R$, we use the notation \mbox{$\mathbfcal{L}=\{\mathcal{L}_x:{\Sin_x\to\Sout_x}\}_{x=1}^L$} to denote a generic \CYthree{channel tuple} in the free set $\mathfrak{F}_R$)}: 
\begin{result}\label{Thm:UolaTheorem} {\em\cite{Uola2020PRA}}
Consider a given ${\bm\Lambda}=\{\Sin_x\to\Sout_x\}_{x=1}^L$.
Let $\mathfrak{F}_{R}$ be convex and compact. 
Then the following two statements are equivalent:
\begin{enumerate}
\item $\mathbfcal{E}\notin\mathfrak{F}_{R}$. That is, $\mathbfcal{E}$ is $R$-resourceful.
\item There exist state ensembles $\{p_{i|x},\rho_{i|x}\}_{i,x}$ with $\sum_{i}p_{i|x} = 1\;\forall\;x$ \CY{and $p_{i|x}\ge0\;\forall\,i,x$}, POVMs $\{E_{j|x}\}_{j,x}$ with $\sum_{j}E_{j|x} = \id_{\Sout_x}\;\forall\,x$ and $E_{j|x}\ge0\;\forall\,j,x$, and real numbers $\{\omega_{ij|x}\}_{i,j,x}$ such that
\begin{align}\label{Eq:Thm1}
&\sum_{i,j,x}\omega_{ij|x}p_{i|x}{\rm tr}\left[E_{j|x}{\mE_x}(\rho_{i|x})\right]
\notag\\&\quad
>\max_{\mathbfcal{L}\in\mathfrak{F}_{R}}\sum_{i,j,x}\omega_{ij|x}p_{i|x}{\rm tr}\left[E_{j|x}{\mathcal{L}_x}(\rho_{i|x})\right].
\end{align}
\end{enumerate}
\end{result}
\begin{proof}
Ref.~\cite{Uola2020PRA} has proved the case when $\Sin_x=\Sin$ and $\Sout_x=\Sout$ for every $x$.
In the general case, it suffices to use the following map onto a set satisfying the conditions imposed by Ref.~\cite{Uola2020PRA}:
\begin{align}
\mathcal{M}:\mathbfcal{E} &= \left\{{\mE_x:{\Sin_x\to\Sout_x}}\right\}_{x=1}^L\nonumber\\
&\mapsto\left\{\wt{\mE}_x = \frac{\id_{\Sout\setminus\Sout_x}}{d_{\Sout\setminus\Sout_x}}\otimes\left(\mE_x\circ{\rm tr}_{\Sin\setminus\Sin_x}\right):\Sin\to\Sout\right\}_{x=1}^L.
\end{align}
As an embedding of the channels {$\mE_x:{\Sin_x\to\Sout_x}$} into a bigger input-output route $\Sin\to\Sout$, the map $\mathcal{M}$ is one-to-one and continuous \CY{(in the topology induced by the diamond norm~\cite{Watrous-book}).}
Hence, the set $\mathcal{M}\left(\mathfrak{F}_{R}\right)$ is convex and compact if $\mathfrak{F}_{R}$ is so.
It suffices to show the sufficiency ($\Rightarrow$) since the converse direction holds immediately.
When $\mathbfcal{E}\notin\mathfrak{F}_{R}$, it means \mbox{$\mathcal{M}\left(\mathbfcal{E}\right)\notin\mathcal{M}\left(\mathfrak{F}_{R}\right)$} since the map $\mathcal{M}$ is one-to-one.
\CY{By Ref.~\cite{Uola2020PRA},} there exist state ensembles $\{p_{i|x},\rho_{i|x}\}_{i,x}$ in $\Sin$, POVMs $\{E_{j|x}\}_{j,x}$ in $\Sout$, and real numbers $\{\omega_{ij|x}\}_{i,j,x}$ such that
\begin{align}
&\sum_{i,j,x}\omega_{ij|x}p_{i|x}{\rm tr}\left[\wt{E}_{j|x}{\mE_x}\left(\wt{\rho}_{i|x}\right)\right]\nonumber\\
&\quad\quad\quad>\max_{\mathbfcal{L}\in\mathfrak{F}_{R}}\sum_{i,j,x}\omega_{ij|x}p_{i|x}{\rm tr}\left[E_{j|x}\wt{\mathcal{L}}_{x}(\rho_{i|x})\right]\nonumber\\
&\quad\quad\quad=\max_{\mathbfcal{L}\in\mathfrak{F}_{R}}\sum_{i,j,x}\omega_{ij|x}p_{i|x}{\rm tr}\left[\wt{E}_{j|x}{\mathcal{L}_x}\left(\wt{\rho}_{i|x}\right)\right],
\end{align}
where we define, for every $i,j,x$,
\begin{align}
\wt{E}_{j|x}\coloneqq\frac{{\rm tr}_{\Sout\setminus\Sout_x}(E_{j|x})}{d_{\Sout\setminus\Sout_x}}\quad\&\quad\wt{\rho}_{i|x}\coloneqq{\rm tr}_{\Sin\setminus\Sin_x}(\rho_{i|x}),
\end{align}
which are, again, POVMs and states.
\end{proof}
\CY{Theorem~\ref{Thm:UolaTheorem} provides a general way to certify quantum resources possessed by \CYthree{channel tuples}, which covers, but not limited to, resources in broadcast scenarios, network-like structures, and channel marginal problems. In general, as long as we have a multipartite setting where we want to study resourcefulness from some local dynamics, Theorem~\ref{Thm:UolaTheorem} can be a potential tool to witness the resource.}

\subsection{Ensemble State Discrimination Tasks as a URCC}
Using Theorem~\ref{Thm:UolaTheorem}, we can then prove a URCC for \CYthree{channel tuple} resources in terms of the {\em ensemble state discrimination task} defined in Ref.~\cite{Hsieh2022PRR}.
Again, consider a given input system $\Sin$, an output system $\Sout$, and a finite set \mbox{${\bm\Lambda}=\{\Sin_x\to\Sout_x\}_{x=1}^L$}.
Then, the task is characterised by 
\begin{align}
D\coloneqq(\{p_x\}_x,\{q_{i|x},\rho_{i|x}\}_{i,x},\{E_{i|x}\}_{i,x}), 
\end{align}
consisting of a probability distribution $\{p_x\}_x$, finitely many state ensembles $\{q_{i|x},\rho_{i|x}\}_{i,x}$ (i.e., $\sum_iq_{i|x}=1\;\forall\,x$ and $q_{i|x}\ge0\;\forall\,i,x$), and POVMs $\{E_{i|x}\}_{i,x}$ (with $\sum_iE_{i|x} = \id_{\Sout_x}\;\forall\,x$ and $E_{i|x}\ge0\;\forall\,i,x$).
Operationally, it means that with probability $p_x$, the input-output route $\Sin_x\to\Sout_x$ is chosen.
Conditioned on this specific $x$, the agent needs to discriminate between the states $\rho_{i|x}$ that are sent with probability $q_{i|x}$ and will be measured by the measurement $\{E_{i|x}\}_i$.
To enhance the discrimination performance, the agent is allowed to use a channel $\mathcal{E}_x:{\Sin_x\to\Sout_x}$ to process the state before the measurement, resulting in the following figure of merit, which is the average success probability (see Ref.~\cite{Hsieh2022PRR} for details):
\begin{align}\label{Eq:P succ}
P\left(D,\mathbfcal{E}\right)\coloneqq\sum_{i,x}p_xq_{i|x}{\rm tr}[E_{i|x}{\mathcal{E}_x}(\rho_{i|x})].
\end{align}
The task $D$ is further called {\em positive} if $p_x~>~0$, $q_{i|x}~>~0$, $E_{i|x}>0$ $\forall\,i,x$.
Our \CY{first main result} generalises Ref.~\cite{Hsieh2022PRR}'s findings, showing that \CY{the above} tasks form a URCC:

\CY{
\begin{result}\label{Result:ImplicationDynamicalResourceTheory} 
Consider a given ${\bm\Lambda}=\{\Sin_x\to\Sout_x\}_{x=1}^L$.
Let $\mathfrak{F}_{R}$ be convex and compact.
Then the following two statements are equivalent:
\begin{enumerate}
\item $\mathbfcal{E}\notin\mathfrak{F}_{R}$. That is, $\mathbfcal{E}$ is $R$-resourceful.
\item There exists a positive ensemble state discrimination task $D$ such that
\begin{align}
P\left(D,\mathbfcal{E}\right) > \max_{\mathbfcal{L}\in\mathfrak{F}_{R}}P\left(D,\mathbfcal{L}\right).
\end{align}
\end{enumerate}
\end{result}
\begin{proof}
We follow the approach of Theorem~2's proof in Ref.~\cite{Hsieh2022PRR}.
It suffices to show that statement 1 implies statement 2 (as statement 2 directly implies statement 1).
Suppose statement 1 holds, i.e., $\mathbfcal{E}\notin\mathfrak{F}_{R}$. 
By Theorem~\ref{Thm:UolaTheorem} (which is applicable as $\mathfrak{F}_{R}$ is convex and compact), there exist state ensembles $\{p_{i|x},\rho_{i|x}\}_{i,x}$ with $\sum_{i}p_{i|x} = 1\;\forall\;x$ and $p_{i|x}\ge0$ $\forall\,i,x$, POVMs $\{E_{j|x}\}_{j,x}$ with $\sum_{j}E_{j|x} = \id_{\Sout_x}\;\forall\,x$ and \mbox{$E_{j|x}\ge0\;\forall\,j,x$,} and real numbers $\{\omega_{ij|x}\}_{i,j,x}$ such that Eq.~\eqref{Eq:Thm1} holds.
Now, for every $i,x$, define the operators 
\begin{align}
L_{i|x}\coloneqq\sum_j\omega_{ij|x}p_{i|x}E_{j|x},
\end{align} 
which are Hermitian but not necessarily positive semi-definite.
Then, using Eq.~\eqref{Eq:Thm1}, we have
\begin{align}
\label{Eq:compu ineq 001}
\sum_{i,x}{\rm tr}\left[L_{i|x}{\mE_x}(\rho_{i|x})\right]>\max_{\mathbfcal{L}\in\mathfrak{F}_{R}}\sum_{i,x}{\rm tr}\left[L_{i|x}{\mathcal{L}_x}(\rho_{i|x})\right],
\end{align}
Now, let $\delta>0$ be a strictly positive value, and let $|{\bf i}|$ and $|{\bf x}|$ be the cardinalities of $i$ and $x$, respectively (i.e., $i=1,...,|{\bf i}|$ and $x=1,...,|{\bf x}|$).
Since $\mathcal{E}_x$'s and $\mathcal{L}_x$'s are channels, we have
\begin{align}
|{\bf i}||{\bf x}|\delta = \sum_{i,x}\delta{\rm tr}\left[{\mE_x}(\rho_{i|x})\right] = \sum_{i,x}\delta{\rm tr}\left[{\mathcal{L}_x}(\rho_{i|x})\right],
\end{align}
which is a strictly positive constant. 
Since adding this constant on both sides of Eq.~\eqref{Eq:compu ineq 001} will preserve the strict inequality, we conclude that Eq.~\eqref{Eq:compu ineq 001} holds if and only if
\begin{align}
&\sum_{i,x}{\rm tr}\left[\left(L_{i|x}+\delta\id_{S'_x}\right){\mE_x}(\rho_{i|x})\right]\nonumber\\\label{Eq:compu ineq 002}
&\quad>\max_{\mathbfcal{L}\in\mathfrak{F}_{R}}\sum_{i,x}{\rm tr}\left[\left(L_{i|x}+\delta\id_{S'_x}\right){\mathcal{L}_x}(\rho_{i|x})\right],
\end{align}
and $\delta>0$ can be arbitrarily chosen.
By choosing
\begin{align}\label{Eq: delta range}
\delta>\max_{i,x}\norm{L_{i|x}}_\infty,
\end{align}
we can ensure that 
$
L_{i|x}+\delta\id_{S'_x} >0
$
$
\forall\,i,x.
$
For any $\delta>0$ in the range Eq.~\eqref{Eq: delta range}, we now further define the following operators
\begin{align}
\widetilde{L}_{i|x}^{(\delta)}\coloneqq\frac{L_{i|x}+\delta\id_{S'_x}}{\CYthree{(1+\delta)}\sum_x\norm{\CYthree{\sum_{i=1}^{|{\bf i}|}}\left(L_{i|x}+\delta\id_{S'_x}\right)}_\infty}\quad\forall\,i,x.
\end{align}
Then, $\widetilde{L}_{i|x}^{(\delta)}>0$ $\forall\,i,x$ and \CYthree{$\sum_{i=1}^{|{\bf i}|}\widetilde{L}_{i|x}^{(\delta)}<\id_{S'_x}$} $\forall\,x$, meaning that $\{\widetilde{L}_{i|x}^{(\delta)}\}_{i=1}^{|{\bf i}|}$ (with a fixed $x$) is {\em part of} a POVM.
Since $\sum_x\norm{\CYthree{\sum_{i=1}^{|{\bf i}|}}\left(L_{i|x}+\delta\id_{S'_x}\right)}_\infty>0$ is constant in $i$ and $x$, dividing it on both sides of Eq.~\eqref{Eq:compu ineq 002} gives
\begin{align}\label{Eq:compu ineq 003}
\sum_{i,x}{\rm tr}\left[\widetilde{L}_{i|x}^{(\delta)}{\mE_x}(\rho_{i|x})\right]>\max_{\mathbfcal{L}\in\mathfrak{F}_{R}}\sum_{i,x}{\rm tr}\left[\widetilde{L}_{i|x}^{(\delta)}{\mathcal{L}_x}(\rho_{i|x})\right].
\end{align}
Now, consider an ensemble state discrimination task with $i=1,...,|{\bf i}|+1$ (i.e., with an additional value) and $x=1,...,|{\bf x}|$,
\begin{align}
D^{(\delta,\epsilon)} = \left(\{p_x^{(\delta,\epsilon)}\}_x,\{q_{i|x}^{(\delta,\epsilon)},\rho_{i|x}^{(\delta,\epsilon)}\}_{i,x},\{E_{i|x}^{(\delta,\epsilon)}\}_{i,x}\right),
\end{align}
which is parametrised by $\delta$ and another parameter $0<\epsilon<1$:
\begin{align}
p_x^{(\delta,\epsilon)}& \coloneqq \frac{1}{|{\bf x}|}\quad\forall\,x;\\
q_{i|x}^{(\delta,\epsilon)}&\coloneqq\frac{1-\epsilon}{|{\bf i}|}\quad{\rm if}\;i=1,...,|{\bf i}|\quad\&\quad q_{|{\bf i}|+1|x}^{(\delta,\epsilon)}\coloneqq\epsilon;\\
\rho_{i|x}^{(\delta,\epsilon)}&\coloneqq\rho_{i|x}\quad{\rm if}\;i=1,...,|{\bf i}|\quad\&\quad \rho_{|{\bf i}|+1|x}^{(\delta,\epsilon)}\coloneqq\eta_x;\\
E_{i|x}^{(\delta,\epsilon)}&\coloneqq\widetilde{L}_{i|x}^{(\delta)}\quad{\rm if}\;i=1,...,|{\bf i}|\nonumber\\
\quad\&\quad&E_{|{\bf i}|+1|x}^{(\delta,\epsilon)}\coloneqq\id_{S'_x}-\sum_{i=1}^{|{\bf i}|}\widetilde{L}_{i|x}^{(\delta)}\CYthree{>0}.
\end{align}
Here, $\eta_x$ is an arbitrarily chosen state in $S_x$, and $\{E_{i|x}^{(\delta,\epsilon)}\}_{i=1}^{|{\bf i}|+1}$ is a valid POVM with $|{\bf i}|+1$ elements.
Also, our construction ensures that $D^{(\delta,\epsilon)}$ is positive, as $p_x^{(\delta,\epsilon)}>0$, $q_{i|x}^{(\delta,\epsilon)}>0$, and $E_{i|x}^{(\delta,\epsilon)}>0$ $\forall\,i,x$.
Following Eqs.~(E9) and (E10) in Ref.~\cite{Hsieh2022PRR}, for a set of channels $\mathbfcal{N} = \{\mathcal{N}_x\}_x$ (again with the given ${\bm\Lambda}=\{\Sin_x\to\Sout_x\}_{x=1}^L$), its average success probability Eq.~\eqref{Eq:P succ} reads
\begin{align}
P(D^{(\delta,\epsilon)},\mathbfcal{N}) = \widetilde{P}^{(\delta)}(\mathbfcal{N}) + \epsilon\times\Gamma^{(\delta)}(\mathbfcal{N}),
\end{align}
where
\begin{align}\label{Eq:P tilde}
\widetilde{P}^{(\delta)}(\mathbfcal{N})\coloneqq\frac{1}{|{\bf i}||{\bf x}|}\sum_{x=1}^{|{\bf x}|}\sum_{i=1}^{|{\bf i}|}{\rm tr}\left[\widetilde{L}_{i|x}^{(\delta)}\mathcal{N}_x(\rho_{i|x})\right]
\end{align}
and
\begin{align}
\Gamma^{(\delta)}(\mathbfcal{N})\coloneqq\frac{1}{|{\bf x}|}\sum_{x=1}^{|{\bf x}|}&\bigg\{{\rm tr}\left[\left(\id_{S'_x}-\sum_{i=1}^{|{\bf i}|}\widetilde{L}_{i|x}^{(\delta)}\right)\mathcal{N}_x(\eta_x)\right]\nonumber\\
&\quad-\frac{1}{|{\bf i}|}\sum_{i=1}^{|{\bf i}|}{\rm tr}\left[\widetilde{L}_{i|x}^{(\delta)}\mathcal{N}_x(\rho_{i|x})\right]\bigg\}.
\end{align}
Now, Eqs.~\eqref{Eq:compu ineq 003} and~\eqref{Eq:P tilde} imply
$
\widetilde{P}^{(\delta)}(\mathbfcal{E})>\max_{\mathbfcal{L}\in\mathfrak{F}_{R}}\widetilde{P}^{(\delta)}(\mathbfcal{L}),
$
meaning that $\alpha^{(\delta)}\coloneqq\widetilde{P}^{(\delta)}(\mathbfcal{E})-\max_{\mathbfcal{L}\in\mathfrak{F}_{R}}\widetilde{P}^{(\delta)}(\mathbfcal{L})>0$ is strictly positive.
Combining everything together, we obtain
\begin{align}
&\max_{\mathbfcal{L}\in\mathfrak{F}_{R}}P(D^{(\delta,\epsilon)},\mathbfcal{L})\le\max_{\mathbfcal{L}\in\mathfrak{F}_{R}}\widetilde{P}^{(\delta)}(\mathbfcal{L})+\epsilon\times\max_{\mathbfcal{L}\in\mathfrak{F}_{R}}\Gamma^{(\delta)}(\mathbfcal{L})\nonumber\\
&\quad=\widetilde{P}^{(\delta)}(\mathbfcal{E})-\alpha^{(\delta)}+\epsilon\times\max_{\mathbfcal{L}\in\mathfrak{F}_{R}}\Gamma^{(\delta)}(\mathbfcal{L})\nonumber\\
&\quad=P(D^{(\delta,\epsilon)},\mathbfcal{E})-\alpha^{(\delta)}+\epsilon\times\left[\max_{\mathbfcal{L}\in\mathfrak{F}_{R}}\Gamma^{(\delta)}(\mathbfcal{L})-\Gamma^{(\delta)}(\mathbfcal{E})\right]\nonumber\\
&\quad\eqqcolon P(D^{(\delta,\epsilon)},\mathbfcal{E})-\alpha^{(\delta)}+\epsilon\beta^{(\delta)},
\end{align}
where we define
$
\beta^{(\delta)}\coloneqq\max_{\mathbfcal{L}\in\mathfrak{F}_{R}}\Gamma^{(\delta)}(\mathbfcal{L})-\Gamma^{(\delta)}(\mathbfcal{E}),
$
which is finite since the function $\Gamma^{(\delta)}$ is bounded for all channels.
Since $\alpha^{(\delta)}>0$, if we also have $\beta^{(\delta)}\le0$, then $P(D^{(\delta,\epsilon)},\mathbfcal{E})>\max_{\mathbfcal{L}\in\mathfrak{F}_{R}}P(D^{(\delta,\epsilon)},\mathbfcal{L})$ $\forall\,0<\epsilon<1$.
On the other hand, if $\beta^{(\delta)}>0$, we can choose $\epsilon$ in the interval
\begin{align}\label{Eq: task interval}
0<\epsilon<\min\left\{\frac{\alpha^{(\delta)}}{\beta^{(\delta)}},1\right\}
\end{align}
to achieve $P(D^{(\delta,\epsilon)},\mathbfcal{E})>\max_{\mathbfcal{L}\in\mathfrak{F}_{R}}P(D^{(\delta,\epsilon)},\mathbfcal{L})$.
This means there is a family of positive tasks $D^{(\delta,\epsilon)}$ [i.e., those with parameter $\delta$ that satisfies Eq.~\eqref{Eq: task interval}]
that can achieve the strict inequality $P(D^{(\delta,\epsilon)},\mathbfcal{E})>\max_{\mathbfcal{L}\in\mathfrak{F}_{R}}P(D^{(\delta,\epsilon)},\mathbfcal{L})$.
\end{proof}
We remark that Theorem~\ref{Result:ImplicationDynamicalResourceTheory} generalises one of the main results of Ref.~\cite{Hsieh2022PRR} (more precisely, Theorem 2 in Ref.~\cite{Hsieh2022PRR}), which further shows that the \CYthree{ensemble state} discrimination tasks introduced in Ref.~\cite{Hsieh2022PRR} not only \CYthree{work} for quantum channel marginal problems, but also actually serve as a URCC.
\CYthree{As an application of Theorem~\ref{Result:ImplicationDynamicalResourceTheory}, which we shall demonstrate in the next section, we can construct a thermodynamic URCC from the ensemble state discrimination tasks.
}

}

\section{Certifying quantum resources by work-extraction tasks}\label{Sec:Work extraction}

\CY{\subsection{Energy Storage Enhancement as a Work-Extraction Task}}
Crucially, one can utilise Theorem~\ref{Result:ImplicationDynamicalResourceTheory} to further show that general quantum resources can also be certified by work extraction.
To see this, we recap the work-extraction task introduced in the companion letter~\cite{Companion-letter}.
Consider a finite-dimensional quantum system with Hamiltonian $H$ and a background temperature $0~<~T~<~\infty$. In thermal equilibrium, it is described by the thermal state 
\begin{align}
\gamma = \frac{e^{-H/k_BT}}{{\rm tr}(e^{-H/k_BT})},
\end{align}
where $k_B$ is the Boltzmann constant, and we keep the Hamiltonian/temperature dependence implicit.
Now, suppose the system is prepared in a non-equilibrium state $\rho$. 
Then, one can extract the following optimal amount of work from it~\cite{Skrzypczyk2014NC,Brandao2013PRL}~\footnote{\CYthree{As mentioned in the companion letter~\cite{Companion-letter}, the quantity $W(\rho,H)$ is the average extractable work in the limit of infinitely many copies of $\rho$~\cite{Brandao2013PRL}, or in the limit of infinitely many steps in the protocol considered in Ref.~\cite{Skrzypczyk2014NC}.}}:
\begin{align}\label{Eq:W}
W(\rho,H) \coloneqq (k_BT\ln2)D\left(\rho\,\|\,\gamma\right),
\end{align}
where $D(\rho\,\|\,\sigma)\coloneqq{\rm tr}\left[\rho\left(\log_2\rho - \log_2\sigma\right)\right]$ is the \CY{quantum (Umegaki) relative entropy~\cite{Umegaki1962}.
When $H=0$, Eq.~\eqref{Eq:W} gives the optimal work extractable from $\rho$'s information content:}
\begin{align}
W_{{\rm inf}}(\rho) \coloneqq W(\rho,H=0) = (k_BT\ln2)D\left(\rho\,\|\,\id/d\right).
\end{align}
\CY{Inspired by \CYthree{Refs.}~\cite{Oppenheim2002PRL,Alimuddin2020PRE,Alimuddin2019PRA,Vaccaro2008PRA,Kwon2018PRL}, we consider the following figure of merit in the companion letter~\cite{Companion-letter}, termed {\em energy storage enhancement} of the state $\rho$ with the Hamiltonian $H$:}
\begin{align}\label{Eq:Work diff}
\Wdiff(\rho,H)\coloneqq W(\rho,H) - W_{{\rm inf}}(\rho).
\end{align}
\CY{As mentioned in the companion letter~\cite{Companion-letter}, $\Wdiff$ carries a rather simple physical meaning---when we prepare the system in $\rho$, the value $\Wdiff(\rho,H)$ describes the ``enhancement'' of energy storage allowed by quenching the Hamiltonian from $0$ (i.e., the fully degenerate one) to $H$. See also Fig.~2 in the companion letter~\cite{Companion-letter}.}
\CY{From a purely operational perspective}, $\Wdiff$ measures the difference between two types of work extraction: One is subject to the given Hamiltonian $H$, and one addresses solely the information content of $\rho$. 
Hence, \CY{to evaluate it for} a given state $\rho$, one simply runs two different work extraction experiments (e.g., by using the protocol introduced in Ref.~\cite{Skrzypczyk2014NC}): One with the given Hamiltonian $H$, and another one with a vanishing Hamiltonian.
After that, one computes the difference between the amounts of the extracted work.

\CY{
\subsection{Work Extraction Game for \CYthree{Channel Tuples}}
As we demonstrate now, $\Wdiff$ can certify general quantum resources of \CYthree{channel tuples}. 
Our strategy is to use $\Wdiff$ to construct a work extraction game for \CYthree{channel tuples} in which quantum \CYthree{resources} can be witnessed.
Let us start with a given ${\bm\Lambda}=\{\Sin_x\to\Sout_x\}_{x=1}^{|\bf x|}$ and a background temperature $0<T<\infty$.
Suppose $\mathfrak{F}_R$ is the free set of a given resource $R$. 
For every label $x$, we pick finitely many initial states $\rho_{i|x}$'s in $\Sin_x$ and Hamiltonians $H_{i|x}$'s of $\Sout_x$, both indexed by $i=1,...,|{\bf i}|$.
Collectively, we write a {\em game} as:
\begin{align}
G\coloneqq\left(\{\rho_{i|x}\}_{i,x},\{H_{i|x}\}_{i,x}\right).
\end{align} 
A fixed game $G$ defines a multi-round work-extraction task as follows.
Consider a \CYthree{channel tuple} \mbox{$\mathbfcal{E}=\{\mathcal{E}_x:\Sin_x\to\Sout_x\}_{x=1}^{|\bf x|}$}. 
In each round, an index $x$ is chosen randomly with equal probability $1/|{\bf x}|$.
After that, a pair of initial state and Hamiltonian from $G$, i.e., $(\rho_{i|x},H_{i|x})$, is further selected randomly with equal probability $1/|{\bf i}|$.
Then, we run work extraction to evaluate $\Wdiff\left[\mathcal{E}_x(\rho_{i|x}),H_{i|x}\right]$~\footnote{\CY{This thus requires several shots of experiments so that we can evaluate both $W_{\rm inf}$ and $W$.}}.
After sufficiently many rounds, the {\em score} of $\mathbfcal{E}$ in the game $G$ reads
\begin{align}\label{Eq:score}
\Delta_{\rm ave}(\mathbfcal{E},G)\coloneqq\frac{1}{|{\bf x}|}\sum_{x}\frac{1}{|\bf i|}\sum_i\Wdiff\left[\mathcal{E}_x(\rho_{i|x}),H_{i|x}\right].
\end{align}
This is $\Wdiff$ of $\mathcal{E}_x(\rho_{i|x})$ with $H_{i|x}$ averaged over $i,x$.
Namely, it is $\mathbfcal{E}$'s average performance in obtaining $\Delta$ with the setting specified by $G$.
Physically, this is the amount of energy storage enhancement $\Delta$ that can be {\em further enhanced} by the given \CYthree{channel tuple} $\mathbfcal{E}$ in the given game $G$.

Finally, we say $G=\left(\{\rho_{i|x}\}_{i,x},\{H_{i|x}\}_{i,x}\right)$ has positive Hamiltonians if $H_{i|x}>0$ $\forall\,i,x$.
Now, we can show that} such work-extraction tasks form a URCC:
\begin{result}\label{AppResult:ThermoOperationalInterpretation}
Consider a given ${\bm\Lambda}=\{\Sin_x\to\Sout_x\}_{x=1}^L$ and temperature $0<T<\infty$.
Let $\mathfrak{F}_{R}$ be convex and compact. 
Then the following two statements are equivalent:
\begin{enumerate}
\item $\mathbfcal{E}\notin\mathfrak{F}_{R}$. That is, $\mathbfcal{E}$ is $R$-resourceful.
\item There exist finitely many Hamiltonians $\{H_{i|x}>0\}_{i,x}$ and states $\{\rho_{i|x}\}_{i,x}$ such that
\begin{align}\label{Eq:main result}
\CY{\Delta_{\rm ave}(\mathbfcal{E},G)>\max_{\mathbfcal{L}\in\mathfrak{F}_R}\Delta_{\rm ave}(\mathbfcal{L},G),}
\end{align}
\CY{where the maximisation is taken over all free \CYthree{channel tuples} denoted by $\mathbfcal{L}\in\mathfrak{F}_R$.}
\end{enumerate}
\end{result}
\begin{proof}
It suffices to show that statement 1 implies \mbox{statement 2}.
Suppose \mbox{$\mathbfcal{E}\notin\mathfrak{F}_{R}$.}
Then, with the listed assumptions, Theorem~\ref{Result:ImplicationDynamicalResourceTheory} implies that there is a positive ensemble state discrimination task $D\coloneqq(\{p_x\}_x,\{q_{i|x},\rho_{i|x}\}_{i,x},\{E_{i|x}\}_{i,x})$ such that
\begin{align}\label{Eq:computation proof thermo witness 001}
P\left(D,\mathbfcal{E}\right) > \max_{\mathbfcal{L}\in\mathfrak{F}_{R}}P\left(D,\mathbfcal{L}\right).
\end{align}
Define
\begin{align}\label{Eq: computation proof Hamiltonian def}
H_{i|x}\coloneqq(k_BT\ln2)p_xq_{i|x}E_{i|x},
\end{align}
which is a Hamiltonian in $\Sout_x$, and it is strictly positive as long as $T>0$.
Combining Eqs.~\eqref{Eq:P succ},~\eqref{Eq:computation proof thermo witness 001}, and~\eqref{Eq: computation proof Hamiltonian def} gives
\begin{align}\label{AppEq:CompWorkExtractionDynamicalResource}
&\sum_{i,x}{\rm tr}[H_{i|x}\mathcal{E}_x(\rho_{i|x})] > \max_{\mathbfcal{L}\in\mathfrak{F}_{R}}\sum_{i,x}{\rm tr}[H_{i|x}\mathcal{L}_x(\rho_{i|x})],
\end{align}
Finally, we note the following formula
\begin{align}\label{Eq:WgapRelation}
&\frac{\Wdiff\left(\eta,H_{i|x}\right)}{k_BT\ln2} = \frac{{\rm tr}\left(H_{i|x}\eta\right)}{k_BT\ln2} + \log_2{\rm tr}\left(e^{-\frac{H_{i|x}}{k_BT}}\right) - \log_2d_{x},
\end{align}
where $d_x$ is the dimension of $\Sout_x$.
Then, adding the constant 
\begin{align}
(k_BT\ln2)\sum_{i,x}\left[\log_2{\rm tr}\left(e^{-\frac{H_{i|x}}{k_BT}}\right) - \log_2d_x\right]
\end{align} 
on both sides of Eq.~\eqref{AppEq:CompWorkExtractionDynamicalResource}, which preserves the strict inequality, and using $\Delta_{\rm ave}$'s definition [Eq.~\eqref{Eq:score}] conclude the proof.
\end{proof}

\CY{
Theorem~\ref{AppResult:ThermoOperationalInterpretation} implies that, as an answer to our central question, the thermodynamic tasks characterised by $\Wdiff$ [Eq.~\eqref{Eq:Work diff}] serve as the first thermodynamic URCC.
\CYthree{Notably, our proof relies on the mathematical link between these thermodynamic tasks [Eq.~\eqref{Eq:score}] and the ensemble state discrimination tasks (which form a URCC via Theorem~\ref{Result:ImplicationDynamicalResourceTheory}) through the formula Eq.~\eqref{Eq:WgapRelation}.
In other words, the thermodynamic URCC that we construct via Eq.~\eqref{Eq:Work diff} is induced by the URCC from the ensemble state discrimination tasks.
It remains an open interesting problem to explore whether the technique adopted here can also provide other discrimination tasks (e.g., Theorem 4 in Ref.~\cite{Hsieh2023-2} and Result 5 in Ref.~\cite{Hsieh2026PRA}) with a novel thermodynamic meaning in work extraction.}

Interestingly, together with the companion letter~\cite{Companion-letter}, the energy storage enhancement $\Delta$ reveals (at least) two different types of quantum resources: For states, it is about the energy that they can {\em carry}; for channels (in \CYthree{channel tuples}), it is about the energy that they can {\em enhance}.
We refer the readers to the companion letter~\cite{Companion-letter} for further details as well as $\Delta$'s applications to characterise state conversions. 
\CYthree{Finally, as discussed} in the companion letter~\cite{Companion-letter}, since $\Wdiff$ can be naturally interpreted as an enhancement of energy storage, Theorem~\ref{AppResult:ThermoOperationalInterpretation} also \CYthree{hints} that general quantum resources \CYthree{of channel tuples} can be \CYthree{potentially} witnessed by quantum batteries.
\CYthree{Further explorations are left for future studies.}

}

\CY{
\subsection{Implications of Theorem~\ref{AppResult:ThermoOperationalInterpretation}}
When one considers specific settings, simpler versions can be obtained. 
First, when $|{\bf x}|=1$, we have a result for channel resources.
In this case, a game $G$ becomes a collection of the form $(\{\rho_i\}_i,\{H_i\}_i)$, and we have
\begin{result-coro}\label{coro:channel}
Let $\mathfrak{F}_R$ be a convex and compact set of channels. 
Then $\mathcal{E}\notin\mathfrak{F}_R$ if and only if there exist states $\{\rho_i\}_i$ and Hamiltonians $\{H_i>0\}_i$ such that
\begin{align}
\frac{1}{|{\bf i}|}\sum_{i}\Wdiff\left[\mathcal{E}(\rho_{i}),H_{i}\right]>\max_{\mathcal{L}\in\mathfrak{F}_R}\frac{1}{|{\bf i}|}\sum_{i}\Wdiff\left[\mathcal{L}(\rho_{i}),H_{i}\right].
\end{align}
\end{result-coro}
In this sense, work extraction advantages can be provided by general channel resources from, e.g., quantum memory~\cite{Rosset2018PRX,Ku2022PRXQ,Yuan2021npjQI}, communication~\cite{Takagi2020PRL,Hsieh2021PRXQ}, and thermodynamics~\cite{Hsieh2021PRXQ,Hsieh2025PRL,Hsieh2025PRA,Stratton2023}.
Furthermore, when we focus on state-preparation channels, we obtain a result for state resources, $R_{\rm state}$, where a game $G$ becomes a {\em single} Hamiltonian:
\begin{result-coro}\label{coro:state}
Let $\mathfrak{F}_{R_{\rm state}}$ be a convex and compact set of states. 
Then \mbox{$\rho\notin\mathfrak{F}_{R_{\rm state}}$} if and only if there exists a Hamiltonian $H>0$ such that
\begin{align}
\Wdiff\left(\rho,H\right)>\max_{\eta\in\mathfrak{F}_{R_{\rm state}}}\Wdiff\left(\eta,H\right).
\end{align}
\end{result-coro}
\begin{proof}
To start with, let
$
\mathfrak{F}_R\coloneqq\{\eta{\rm tr}(\cdot)\,|\,\eta\in\mathfrak{F}_{R_{\rm state}}\},
$
which is the set of all state-preparation channels that output some state in $\mathfrak{F}_{R_{\rm state}}$. This set is convex and compact when $\mathfrak{F}_{R_{\rm state}}$ is so. Using Eq.~\eqref{AppEq:CompWorkExtractionDynamicalResource} for a state-preparation channel $\mathcal{E}(\cdot) = \rho{\rm tr}(\cdot)$ and applying Eq.~\eqref{Eq:WgapRelation} completes the proof.
\end{proof}
Hence, here, we provide an alternative proof of Theorem 3 in the companion letter~\cite{Companion-letter}.
}

\section{Certifying entanglement by one-sided device-independent work-extraction tasks}\label{Sec:1SDI}
From Theorem~\ref{AppResult:ThermoOperationalInterpretation}, it is clear that we can use work extraction to certify state resources, such as entanglement.
This, however, requires trusted devices; namely, we need to fully control our work extraction system.
It is interesting to notice that the technique adopted in proving Theorem~\ref{AppResult:ThermoOperationalInterpretation} also implies that entanglement can be certified by work-extraction tasks in a {\em one-sided device-independent} way.
More precisely, to certify the entanglement of a bipartite state, we can show that it suffices to extract work from {\em one side} of it without trusting the other party's operations.

\CY{\subsection{Quantum Steering: A Recap}}
To formalise our result, we briefly recap the notion of {\em quantum steering}, \CY{or simply {\em steering}}~\cite{EPR,Schrodinger1935,Schrodinger1936,Wiseman2007PRL,Jones2007PRA,UolaRMP2020,Cavalcanti2016}, \CY{which is a multi-round} scenario with two agents $A$ and $B$.
\CY{Suppose they share a bipartite state $\rho_{AB}$.
Then, the goal of the {\em steering scenario} is for the agent $B$ to certify that $\rho_{AB}$ is entangled {\em without} trusting the agent $A$ after all.
To define a steering scenario, the first thing for $B$ to do is to ask $A$ to prepare a collection of quantum measurements locally in $A$'s lab. Mathematically, this is described by ${\bf E}=\{E_{a|x}\}_{a,x}$, where for every fixed $x$, the set $\{E_{a|x}\}_a$ is a POVM with outcomes labelled by $a$.
We call the collection ${\bf E}$ a {\em measurement assemblage}.

Once a measurement assemblage ${\bf E}$ is specified, the agents can start the steering scenario.}
\CY{During each round, first,} agent $A$ selects a measurement setting $x$ \CY{uniformly.} \CY{Then, $A$ performs the corresponding POVM $\{E_{a|x}\}_a$ from ${\bf E}$, obtaining the outcome $a$.}
\CY{After that, $A$ informs $B$} about the measurement setting and \CY{outcome, indicated by the index pair} $(a,x)$. 
\CY{For a given $x$ (and note that each $x$ has equal probability to occur), the pair $(a,x)$} occurs with \CY{the conditional} probability 
\begin{align}
\CY{P(a|x) = {\rm tr}[(E_{a|x}\otimes\id_B)\rho_{AB}]}
\end{align}
and results in the conditional quantum state 
\begin{align}
\CY{\norsigma_{a|x}={\rm tr}_A[(E_{a|x}\otimes\id_B)\rho_{AB}]/P(a|x)}
\end{align}
\CY{in $B$'s system.} 
\CY{Overall, the final statistics and states that $B$ obtains can be collectively described by} the so-called {\em state assemblage}, \CY{which is} \CY{denoted by} \mbox{$\assem=\{\assemsigma_{a|x}\}_{a,x}$}~\cite{PhysRevA.88.032313} with 
\begin{align}\label{Eq:SA}
\assemsigma_{a|x}\coloneqq P(a|x)\norsigma_{a|x}\CY{={\rm tr}_A[(E_{a|x}\otimes\id_B)\rho_{AB}]}.
\end{align}
\CY{$\assem$ fully describes the steering scenario, as it collects the ``data,'' including post-measurement quantum states and classical statistics, that \CYthree{$B$} will obtain after sufficiently many rounds.}
\CY{Note that, by construction of Eq.~\eqref{Eq:SA}, we have} 
\begin{align}
\sum_a\assemsigma_{a|x} = \sum_a\assemsigma_{a|y}\eqqcolon\rho_B\quad\forall\,x,y,
\end{align}
\CY{which is usually called} the no-signalling condition for steering, where $\rho_B$ is the reduced density matrix for $B$. 
\CY{Physically, this means that $A$ cannot change $B$'s average quantum state by changing the measurement setting $x$; namely, $A$ cannot ``signal'' $B$'s local quantum output by their local classical input.}

\CY{Crucially, as mentioned earlier, the steering scenario aims to help $B$ to certify $\rho_{AB}$'s entanglement. Hence, let us now examine what would happen if $\rho_{AB}$ is not entangled; i.e., separable.
In that case, we can write $\rho_{AB}$ as
\begin{align}
\CY{\rho_{AB}^{(\rm sep)} = \sum_\lambda p(\lambda)\rho_A^{(\lambda)}\otimes\rho_B^{(\lambda)}}
\end{align}
with some probability distribution $\{p(\lambda)\}_\lambda$ and variable $\lambda$.
Then, \CY{$B$'s} obtained state assemblage, indicated by \mbox{$\assem^{(\rm sep)}\coloneqq\{\assemsigma^{(\rm sep)}_{a|x}\}_{a,x}$}, must take the form
\begin{align}
\assemsigma_{a|x}^{(\rm sep)}\stackrel{\rm LHS}{=}\sum_\lambda P(a|x,\lambda)p(\lambda)\rho_B^{(\lambda)}\quad\forall\,a,x,
\end{align}
where $P(a|x,\lambda)={\rm tr}(E_{a|x}\rho_A^{(\lambda)})$.
This is the so-called {\em local hidden-state} (LHS) model ~\cite{EPR,Schrodinger1935,Schrodinger1936,Wiseman2007PRL,Jones2007PRA,UolaRMP2020,Cavalcanti2016}.
Let ${\bf LHS}$ denote the set of all state assemblages that admit some LHS representation.
An important physical implication is that, \CY{if $B$ obtains} a state assemblage $\assem$ that {\em cannot} be described by {\em any} LHS model (i.e., $\assem\notin{\bf LHS}$, which is called {\em steerable}), then, immediately, $B$ can conclude that the state $\rho_{AB}$ {\em must} be entangled.
Crucially, as all $B$ needs to do is to examine ingredients in $\assem$, which are all locally obtainable in $B$'s lab, the agent $B$ can certify $\assem\notin{\bf LHS}$ {\em without} trusting $A$ after all---this is the {\em one-sided} device-independent certification of entanglement~\cite{UolaRMP2020,Cavalcanti2016}, which is also the foundation of one-sided device-independent quantum information processing~\cite{Branciard2012PRA,Piani2015,PhysRevX.5.041008}.}

\CY{Based on the above discussion, for steering, we thus have} the free set $\mathfrak{F}_{\rm steering}={\bf LHS}$.
\CY{Any state assemblage that is outside this set are said to demonstrate quantum steering.}
Importantly, instead of checking the whole set ${\bf LHS}$ to certify steering, it suffices to check $\assem\notin{\bf LHS}(\assem)$, where ${\bf LHS}({\assem})$ denotes the set of state assemblages \mbox{$\assemLHS = \{\assemtau_{a|x}\}_{a,x}\in{\bf LHS}$} satisfying ${\rm tr}(\assemtau_{a|x}) = \CY{{\rm tr}(\assemsigma_{a|x})}$ for every $a,x$. Physically, these are LHS state assemblages that are indistinguishable from $\assem$ based on the classical statistics that $B$ receives from $A$.
\CY{To show the claim,} first, note that $\assem\in{\bf LHS}$ implies that $\assem\in{\bf LHS}(\assem)$.
On the other hand, since ${\bf LHS}(\assem)\subseteq{\bf LHS}$, $\assem\notin{\bf LHS}$ implies that $\assem\notin{\bf LHS}(\assem)$.
Hence, we have
\begin{align}\label{Eq:useful lemma}
\assem\notin{\bf LHS}\quad\text{\em if and only if}\quad\assem\notin{\bf LHS}(\assem),
\end{align}
\CY{and it suffices to check $\assem\notin{\bf LHS}(\assem)$ to achieve one-sided device-independent entanglement certification.}

\CY{\subsection{One-sided Device-Independent Work-Extraction Tasks}}
We are now in the position to detail the local work-extraction task that only involves \CY{$B$'s system.} Consider a given state assemblage $\assem$ and a given collection of Hamiltonians ${\bf H}\coloneqq\{H_{a|x}\}_{a,x}$ in \CY{$B$'s system}. Upon receiving the information $(a,x)$ about the measurement setting and \CY{outcome} of agent $A$, agent $B$ abruptly changes the local Hamiltonian to $H_{a|x}$ (without changing the quantum state) and then extracts work from the conditional state $\norsigma_{a|x}$. Note that, in contrast to the sufficient steering witnesses of, e.g., Refs.~\cite{PhysRevA.40.913,Yadin2021NC}, here, each pair $(a,x)$ may need a different Hamiltonian. 
With uniformly distributed $x$ [i.e., $P(x)$ is a constant in $x$], we obtain the average extractable work as
\begin{align}
\overline{W}({\assem},{\bf H})\coloneqq\sum_{a,x}P(a,x)W\left(\norsigma_{a|x},H_{a|x}\right),
\end{align}
where \mbox{$P(a,x) = P(x)P(a|x)$.}
Next, the two agents implement the same protocol as before, the only difference being that this time, \CY{$B$'s local Hamiltonian} is always fully degenerate.
This corresponds to the average extractable work
\begin{align}
\overline{W}_{\rm inf}(\mathbfcal{A})\coloneqq\sum_{a,x}P(a,x)W_{\rm inf}\left(\norsigma_{a|x}\right).
\end{align}
We again consider the deficit-type figure-of-merit
\begin{align}
\overline{\Wdiff}(\mathbfcal{A},{\bf H})&\coloneqq\overline{W}(\mathbfcal{A},{\bf H})-\overline{W}_{\rm inf}(\mathbfcal{A})\nonumber\\
&=\sum_{a,x}P(a,x)\Wdiff\left(\norsigma_{a|x},H_{a|x}\right),
\end{align}
which can be understood as \CY{the energy storage enhancement $\Delta$ [Eq.~\eqref{Eq:Work diff}]} conditioned on the indices $a,x$ from steering scenarios. It turns out that this class of thermodynamic work-extraction tasks is necessary and sufficient to certify quantum steering \CY{and thus relevant to one-sided device-independent entanglement certification}. Specifically, we have \CY{(from now on, we use the notation ${\assemLHS}$ to denote a generic element in the set ${{\bf LHS}(\assem)}$)}
\begin{result}\label{Result:Steering}
${\assem}\notin{\bf LHS}$ if and only if, for every \mbox{$0<T<\infty$}, there exist Hamiltonians \mbox{${\bf H} = \{H_{a|x}\ge0\}_{a,x}$} such that
\begin{align}\label{Eq:result2}
\overline{\Wdiff}({\assem},{\bf H})>\max_{{\assemLHS}{\in{\bf LHS}(\assem)}}\overline{\Wdiff}({\assemLHS},{\bf H}).
\end{align}
\end{result}
\begin{proof}
Since $\assem\notin{\bf LHS}$ if and only if $\assem\notin{\bf LHS}(\assem)$, \CY{as in Eq.~\eqref{Eq:useful lemma},} it suffices to show that $\assem\notin{\bf LHS}(\assem)$ implies the desired strict inequality for every given $0<T<\infty$.
Consider one such state assemblage $\assem$.
Then, by the dual programme of steering robustness (see, e.g., Eq.~(41) in Ref.~\cite{Cavalcanti2016}, Eq.~(8) in Ref.~\cite{Piani2015}, Eq.~(28) in Ref.~\cite{Hsieh2016}), there exists a collection of semi-definite positive operators $\{F_{a|x}\ge0\}_{a,x}$ such that
\begin{align}\label{Eq:Compu01}
\CY{\sum_{a,x}{\rm tr}\left(F_{a|x}\assemsigma_{a|x}\right) > \max_{{\assemLHS}\in{\bf LHS}}\sum_{a,x}{\rm tr}\left(F_{a|x}\assemtau_{a|x}\right).}
\end{align}
\CY{Now, by reducing the maximisation range, we have
\begin{align}
\max_{{\assemLHS}\in{\bf LHS}}\sum_{a,x}{\rm tr}\left(F_{a|x}\assemtau_{a|x}\right)\ge\max_{{\assemLHS}\in{\bf LHS}(\assem)}\sum_{a,x}{\rm tr}\left(F_{a|x}\assemtau_{a|x}\right).
\end{align}
Combining the above two inequalities gives
\begin{align}\label{Eq:Compu02}
\sum_{a,x}{\rm tr}\left(F_{a|x}\assemsigma_{a|x}\right) > \max_{{\assemLHS}\in{\bf LHS}(\assem)}\sum_{a,x}{\rm tr}\left(F_{a|x}\assemtau_{a|x}\right).
\end{align}
}
Note that we need \CY{the above form} to obtain $\overline{\Delta}$ with appropriate $P(a,x)$. 
We define the Hamiltonians
\begin{align}
H_{a|x}\coloneqq(k_BT\ln2)F_{a|x}.
\end{align}
Then, using Eq.~\eqref{Eq:WgapRelation}, and adding the constant 
\begin{align}
\sum_{a,x}P(a,x)\left[\log_2{\rm tr}\left(e^{-\frac{H_{a|x}}{k_BT}}\right) - \log_2d_{x}\right]
\end{align}
to both sides of Eq.~\eqref{Eq:Compu02}, the desired result follows by multiplying $k_BT\ln2$.
\end{proof}

Our finding thus demonstrates that steering is the necessary and sufficient resource for the advantage in \CY{the one-sided device-independent work-extraction tasks introduced above.}
This further complements the recent findings reported in Refs.~\cite{JiPRL2022,BeyerPRL2019,ChanPRA2022,Hsieh2024,Biswas2025PRL}.
Theorem~\ref{Result:Steering} thus implies that entanglement distributed among two agents can be thermodynamically certified by one of them without trusting the other. More precisely, by obtaining the extractable work $\overline{\Wdiff}(\assem,{\bf H})$ (which can be done in $B$ with the classical statistics announced by $A$, e.g., via a public channel), Theorem~\ref{Result:Steering} implies that the local agent $B$ can certify that $\assem\notin{\bf LHS}$ if the obtained work value is strictly higher than $\max_{{\assemLHS}\in{\bf LHS}(\assem)}\overline{\Wdiff}({\assemLHS},{\bf H})$ (which is a value that $B$ can compute).
In that case, the state assemblage $\assem\notin{\bf LHS}$ cannot be produced by any separable state in a steering scenario, and, consequently, $\rho_{AB}$ must be entangled.

\subsection{Application: Local Anomalous Energy Flow Generated by Entanglement and Incompatible Measurements}
As an application, Theorem~\ref{Result:Steering} uncovers a novel type of anomalous energy flow. 
In the literature, anomalous energy flow refers to energy flow that classical systems cannot induce.
For instance, energy mainly flows from the hot to the cold systems in classical thermodynamics, while quantum correlations can reverse this time arrow, leading to the anomalous heat backflow (see, e.g., Refs.~\cite{Lipkabartosik2023,Jennings2010PRE}).
Here, Theorem~\ref{Result:Steering} suggests that a certain type of locally observed energy flow implies the presence of globally distributed entanglement.
To see this, let us define
\begin{align}\label{Eq:anomalous energy flow}
\CY{E_{\rm anomaly}}(\assem)\coloneqq\max_{{\bf H}}\left(\overline{\Wdiff}({\assem},{\bf H})-\max_{{\assemLHS}\in{\bf LHS}(\assem)}\overline{\Wdiff}({{\assemLHS}},{\bf H})\right),
\end{align}
where the maximisation is taken over all possible local Hamiltonians ${\bf H}=\{H_{a|x}\}_{a,x}$ with $0\le H_{a|x}\le\delta\id$, where $\delta>0$ is a given energy scale whose dependence is kept implicit. Then, when $\CY{E_{\rm anomaly}}(\assem)>0$ for some properly designed local Hamiltonians ${\bf H}$ in \CY{$B$'s system, the local agent $B$ concludes} that there must be some entangled state behind the entire scenario---as no separable state can produce this amount of energy.
We thus conclude that $\CY{E_{\rm anomaly}}(\assem)$ is the maximal amount of energy that is certainly induced by entanglement. Importantly, the local agent can certify this novel type of anomalous energy flow without accessing the bipartite structure.
We thus discover a novel type of one-sided device-independent anomalous energy flow (see also Refs.~\cite{BeyerPRL2019,JiPRL2022,ChanPRA2022}).

As another illustrative application of Theorem~\ref{Result:Steering}, let us relax the one-sided device-independent assumption and always consider $A$ and $B$ sharing a maximally entangled state $\ket{\Phi^+}_{AB}\coloneqq\sum_{i=0}^{d-1}(1/\sqrt{d})\ket{ii}_{AB}$. Then, Theorem~\ref{Result:Steering} implies that the local anomalous energy flow $\CY{E_{\rm anomaly}}(\assem)$ is equivalent to {\em measurement incompatibility} in the system $A$. 
\CY{More precisely, a measurement assemblage ${\bf M}=\{M_{a|x}\}_{a,x}$ is said to be {\em jointly measurable}, or {\em compatible}, if there exists a single POVM $\{G_\lambda\}_\lambda$ and some conditional probability distribution $\{P(a|x,\lambda)\}_{a,x,\lambda}$ achieving~\cite{Otfried2021Rev,UolaRMP2020}
\begin{align}
M_{a|x}\stackrel{\rm JM}{=}\sum_\lambda P(a|x,\lambda)G_\lambda\quad\forall\,a,x.
\end{align}
Namely, POVMs in ${\bf M}$ can be simulated by a {\em single} POVM plus some purely classical processing.
A measurement assemblage is called {\em incompatible} if it is not compatible.}
\CY{Now, suppose in a steering scenario $A$ and $B$ already share} a maximally entangled state $\ket{\Phi^+}_{AB}$.
\CY{When $A$ is assigned to perform the measurement assemblage ${\bf E}$, $B$'s state assemblage ${\assem}^+({\bf E})\coloneqq\{\assemsigma^+_{a|x}({\bf E})\}_{a,x}$ takes the form} 
\begin{align}
\CY{
\assemsigma^+_{a|x}({\bf E})\coloneqq{\rm tr}_A\left[(E_{a|x}\otimes\id_B)\proj{\Phi^+}_{AB}\right]\quad\forall\,a,x.
}
\end{align} 
\CY{Crucially,} it is known that \CY{${\assem}^+({\bf E})$} is steerable {\em if and only if} \CY{${\bf E}$} is incompatible~\cite{Otfried2021Rev,UolaRMP2020}. \CY{Using this fact with} Theorem~\ref{Result:Steering} gives
\begin{result-coro}
The measurement assemblage \CY{${\bf E}$} is incompatible if and only if $\CY{E_{\rm anomaly}}\left[{\assem}^+(\CY{{\bf E}})\right]>0$ for some energy scale $\delta>0$.
\end{result-coro}
Hence, we reveal another physical underpinning of the anomalous energy flow $\CY{E_{\rm anomaly}}$, which are incompatible measurements in \CY{$A$'s system}. This result thus establishes a thermodynamic meaning for incompatible measurements---quantum incompatibility and entanglement can produce locally extractable energy flows that cannot be observed in any classical system.

Finally, as a remark, it is clear that any nontrivial steering inequality can induce some Hamiltonians to certify steering via the anomalous energy flow defined in Eq.~\eqref{Eq:anomalous energy flow}. 
This allows us to explicitly construct Hamiltonians for, e.g., the {\em superactivation} of the anomalous energy flow. 
More precisely, superactivation of steering refers to situations where a bipartite quantum state $\rho_{AB}$ cannot demonstrate steering (i.e., state assemblages induced by it are in ${\bf LHS}$), whereas at the same time, $\rho_{AB}^{\otimes k}$ can violate some steering inequalities when $k$ is large enough~\cite{Hsieh2016,QuintinoPRA2016}. Using Theorem~\ref{Result:Steering}, we deduce that such a $\rho_{AB}$ cannot induce any anomalous energy flow defined in Eq.~\eqref{Eq:anomalous energy flow}, while $\rho_{AB}^{\otimes k}$ can do so with a large $k$. Moreover, the corresponding Hamiltonians can be explicitly found by using the steering inequality Eq.~(40) in Ref.~\cite{Hsieh2016} and multiplying it by $(k_BT\ln2)$, as in the proof of Theorem~\ref{Result:Steering}.

\section{Conclusions}\label{Sec:Conclusion}
In this work, we introduce a thermodynamic universal resource certification class based on work extraction, implying that general quantum resources can demonstrate advantages in work-extraction tasks. 
Our approach further shows that locally performed work extraction can certify globally distributed entanglement in a one-sided device-independent way, and a novel type of anomalous energy flow underpinned by entanglement and measurement incompatibility is reported.

We now discuss the physical implications of our findings.
First, our results actually also imply that quantum complementarity~\cite{Hsieh2023} as well as different notions of quantum incompatibility (see, e.g., Refs.~\cite{Otfried2021Rev,Haapasalo2014,Girard2021NC,Viola2015, Haapasalo2021Quantum,Haapasalo2015,Carmeli2019PRL,Carmeli2019JMP,Hsieh2024}) can provide work extraction advantages.
The former can be viewed as a thermodynamic version of Theorem 2 in Ref.~\cite{Carmeli2019PRL} (see also Theorem 1 in Ref.~\cite{Uola2019PRL}), which also simplifies the operational task introduced in Ref.~\cite{Skrzypczyk2019}. 
The latter is related to Corollary 2 in Ref.~\cite{Carmeli2019JMP} and Theorem 1 in Ref.~\cite{Mori2020PRA}.
These findings, however, go beyond the scope of this work, as here we aim to focus on general quantum resources, rather than specific ones.
Consequently, we report on these findings in a follow-up paper focusing on thermodynamic advantages of the uncertainty principle's signature (especially incompatibility and complementarity), which is now in preparation~\cite{Hsieh-preparation}.

Furthermore, as mentioned earlier, it is unclear whether the technique used in this work can also apply to channel discrimination tasks as detailed in Theorem 4 in Ref.~\cite{Hsieh2023-2}.
More precisely, is it possible to turn the channel discrimination tasks into a novel thermodynamic URCC?
Also, can any {\em other} class of thermodynamics tasks serve as URCC?

As another relevant aspect, recently, it has been found that probabilistic tasks may also provide a universal characterisation of general quantum resources via the so-called filters (e.g., Theorem D.1 in Ref.~\cite{Hsieh2025PRA-3}). Notably, it suffices to focus on filters with a relatively simple form (such as
filters with a single Kraus operator; see also, e.g., Refs.~\cite{Ku2022NC,Ku2023,Hsieh2023}). It is then valuable, both theoretically and practically, to investigate the thermodynamic interpretation of such filters and their potential applications to characterising/certifying general quantum resources. More broadly, it is rewarding to extend the recent findings in probabilistic conversions (e.g., Refs.~\cite{Regula2022PRL,Regula2022Quantum}) to thermodynamics.

Finally, it may be insightful to apply our findings to specific resources, in particular, {\em thermodynamic} resources. For instance, when we apply our main finding in the companion letter~\cite{Companion-letter} to the so-called {\em informational non-equilibrium}~\cite{purity-review}, we can reproduce the famous majorisation conditions for unital channel conversions (see, e.g., Ref.~\cite{purity-review}).
Can we obtain more from our findings? 
In particular, recently, an Mpemba-like effect has been reported in the concentration problem of informational non-equilibrium~\cite{Hsieh2025PRA-inf}. 
Can our findings provide further insights in this direction?
We leave all these open questions for future research, and we hope to use this work to initiate the community's interest in the interface of general quantum resources and thermodynamics.

\section*{Acknowledgements}
We thank Antonio Ac\'in, Pharnam Bakhshinezhad, Shin-Liang Chen, Huan-Yu Ku, Patryk Lipka-Bartosik, Matteo Lostaglio, Paul Skrzypczyk, Alexander Streltsov, and Benjamin Stratton for fruitful discussions and comments.
C.-Y.~H. acknowledges support from the Royal Society through Enhanced Research Expenses (NFQI), the ERC Advanced Grant (FLQuant), and the Leverhulme Trust Early Career Fellowship (``Quantum complementarity: a novel resource for quantum science and technologies'' with Grant No.~ECF-2024-310). This work is supported by the project PID2023-152724NA-I00, with funding from MCIU/AEI/10.13039/501100011033 and FSE+, by the project CNS2024-154818 with funding by MICIU/AEI/10.13039/501100011033, by the project RYC2021-031094-I, with funding from MCIN/AEI/10.13039/501100011033 and the European Union ‘NextGenerationEU’ PRTR fund, by the project CIPROM/2022/66 with funding by the Generalitat Valenciana, and by the Ministry of Economic Affairs and Digital Transformation of the Spanish Government through the QUANTUM ENIA Project call—QUANTUM SPAIN Project, by the European Union through the Recovery, Transformation and Resilience Plan—NextGenerationEU within the framework of the Digital Spain 2026 Agenda, and by the CSIC Interdisciplinary Thematic Platform (PTI+) on Quantum Technologies (PTI-QTEP+). This work is supported through the project CEX2023-001292-S funded by MCIU/AEI.

\bibliography{Ref.bib}

\end{document}